\documentclass[11pt]{article}
\usepackage{amsfonts}
\usepackage{amssymb,amsmath}
\usepackage{amsthm}

\usepackage{latexsym}

\usepackage{color}
\usepackage{enumerate}
\usepackage[left=1.0in,top=1.0in,right=1.0in,nohead]{geometry}

\newtheorem{theorem}{Theorem}
\newtheorem{lemma}[theorem]{Lemma}
\newtheorem{corollary}[theorem]{Corollary}

\newtheorem{claim}[theorem]{Claim}

\newtheorem{definition}{Definition}

\newcommand{\pr}{\ensuremath{\operatorname{\mathbf{P}}}}
\newcommand{\E}{\ensuremath{\operatorname{\mathbf{E}}}}
\newcommand{\card}[1]{\ensuremath{\left| #1 \right|}} 
\newcommand{\set}[1]{\ensuremath{\left\{#1\right\}}} 
\newcommand{\floor}[1]{\ensuremath{\lfloor\, #1\,\rfloor}}
\newcommand{\ora}{\ensuremath{\mathcal O}}
\newcommand{\po}{\ensuremath{\mathcal{P}}}
\newcommand{\dom}{\ensuremath{\succ}}
\newcommand{\ndom}{\ensuremath{\nsucc}}
\newcommand{\incomp}{\ensuremath{\not\sim}}
\newcommand{\height}{\ensuremath{\operatorname{height}}}
\newcommand{\C}{\ensuremath{\mathcal C}}
\newcommand{\wid}{\ensuremath{w}}

\newcommand{\domtr}{\ensuremath{\unrhd}}
\newcommand{\A}{\ensuremath{\mathcal A}}
\newcommand{\B}{\ensuremath{\mathcal B}}
\newcommand{\cost}{\ensuremath{\operatorname{cost}}}
\newcommand{\argmin}{\ensuremath{\operatorname*{argmin}}}

\newcommand{\defic}{\ensuremath{\operatorname{def}}}
\newcommand{\DEF}{\ensuremath{\mbox{{\sc def}}}}

\newlength{\pgmtab}  
\newenvironment{program}{%
\begin{tabbing}\hspace{0em}\=\hspace{0em}\=%
\hspace{\pgmtab}\=\hspace{\pgmtab}\=\hspace{\pgmtab}\=\hspace{\pgmtab}\=%
\hspace{\pgmtab}\=\hspace{\pgmtab}\=\hspace{\pgmtab}\=\hspace{\pgmtab}\=%
\+\+\kill}{\end{tabbing}}
\newcommand {\ELSE}{{\bf else\ }}
\newcommand {\IF}{{\bf if\ }}
\newcommand {\THEN}{\mbox{\bf then\ }}
\newcommand {\RETURN}{\mbox{\bf return\ }}

\newcommand{\nposets}{\ensuremath{N_{\wid}(n)}}  
\newcommand{\nposetsAlt}{\ensuremath{N_{2\wid}(n)}}
\newcommand{\chainmerge}{\textsc{ChainMerge}}
\newcommand{\peeling}{\textsc{Peeling}}

\begin{document}
\allowdisplaybreaks[1]
\title{Sorting and Selection in Posets}
\author{Constantinos Daskalakis\footnote{Computer Science, U.C. Berkeley.{\tt email:costis@cs.berkeley.edu}} \and Richard M. Karp\footnote{Computer Science, U.C. Berkeley.{\tt email:karp@icsi.berkeley.edu}} \and
Elchanan Mossel\footnote{Statistics, U.C. Berkeley. {\tt email:mossel@stat.berkeley.edu}. Supported by NSF grants DMS 0528488 and DMS 0548249 (CAREER), ONR and N0014-07-1-05-06 and a Sloan Fellowship in Mathematics} \and Samantha Riesenfeld\footnote{Computer Science, U.C. Berkeley. {\tt email:samr@EECS.Berkeley.EDU}}\and Elad Verbin\footnote{Computer Science, Tel Aviv University. {\tt email:eladv@post.tau.ac.il}} }
\date{July 10, 2007}

\addtocounter{page}{-1}
\maketitle
\begin{abstract}
Classical problems of sorting and  searching assume an underlying
linear ordering of the objects being compared.
In this paper, we study a more general setting,
in which some pairs of objects are incomparable.
This generalization is relevant in
applications related to rankings in sports,
college admissions, or conference submissions. It also
has potential applications in biology, such as comparing the
evolutionary fitness of different strains of bacteria, or understanding
input-output relations among a set of metabolic  reactions or the causal
influences among a set of interacting genes or proteins.
Our results improve and extend results from two decades ago of Faigle and Tur\'{a}n,
who were the first to consider some of the problems considered here.

A poset is defined as a set of elements with a transitive partial
order where some pairs of elements may be incomparable. A measure of
complexity of a poset is given by its {\em width}, which is
the maximum size of a set of mutually incomparable elements. We
consider algorithms that obtain information about a poset by queries
that compare two elements. We consider  two complexity
measures:  query complexity, which counts only the number of
queries, and total complexity, which counts all operations.

We present an algorithm that sorts a width $\wid$ poset of size $n$
and has query complexity $O(n (\wid + \log n))$, which is within a
constant factor of the information-theoretic lower bound.
We also show that a variant of Mergesort has query
complexity $O(\wid n \log \frac{n}{\wid})$ and total complexity
$O(\wid^2 n \log \frac{n}{\wid})$.  Faigle and Tur\'{a}n have shown that
the sorting problem has query complexity
$O\left(\wid n \log \frac{n}{\wid}\right)$ but did not address the total compexity of the problem.

Two problems related to sorting are the problem of finding the minimal elements in a
poset and its generalization of finding the bottom $k$ ``levels'',
called the {\em $k$-selection} problem. We give
efficient deterministic and randomized algorithms for finding the minimal
elements with $O(\wid n)$ query and total complexity. We provide
matching lower bounds for the query complexity up to a factor of $2$
and generalize the results to the $k$-selection problem.
We also derive upper bounds on the total complexity of some other
problems of a similar flavor,
such as computing a linear extension of a poset and computing the heights of all elements.

Many open problems remain, of which the most significant is to
determine the precise total complexity of sorting, as well as the precise query and
total complexity of $k$-selection.
It would also be interesting to find efficient static and dynamic data structures
that play the same role for partial orders that heaps and binary search trees
play for total orders.
\end{abstract}
\thispagestyle{empty}
\newpage

\section{Introduction}
\label{sec:intro}

Sorting is the process of determining the underlying linear ordering
of a set $S$ of $n$ elements. {\it Comparison algorithms}, in which
direct comparisons between pairs of elements of $S$ are the only
means of acquiring information about the linear ordering, form an
important subclass, including such familiar algorithms as Heapsort,
Quicksort, Mergesort, Shellsort and Bubblesort.

In this paper we extend the theory of comparison sorting to the
case where the underlying structure of the set $S$ is a partial
order, in which an element may be larger than, smaller than, or incomparable to
another element, and the ``larger-than'' relation is transitive and irreflexive.
Such a set is called a partially ordered set, or poset.
This extension is applicable to many ranking problems where
certain pairs of elements are incomparable.
Examples include ranking college applicants, conference submissions,
tennis players,
strains of bacteria according to their evolutionary fitness,
and points in $R^d$ under the coordinate-wise dominance
relation.

Our algorithms gather information by {\it queries} to an oracle.
The oracle's response to a query involving elements $x$ and $y$ is either
the relation between $x$ and $y$ or a statement of their incomparability.
In many applications, a query may involve extensive effort (for
example, running an experiment to determine the relative
evolutionary fitness of two strains of bacteria, or comparing the
credentials of two candidates for nomination to a learned society).
We therefore consider two measures of complexity for an algorithm or problem: the {\it query complexity}, which is the number of queries performed, and the {\it total complexity}, which is the number of computational operations of all types performed (basic operations include standard data structure operations involving one or two elements of the poset).

A partial order on a set can be thought of as the reachability
relation of a directed acyclic graph (DAG). More generally, a
transitive relation (which is not necessarily irreflexive) can be
thought of as the reachability relation of a general directed graph.
In applications, the relation
represents the direct and indirect causal influences among a set of
variables, processes, or components of a system.
We show that with negligible overhead,
the problem of sorting a
transitive relation reduces to the problem of sorting a partial
order.
Our algorithms thus allow one to reconstruct general directed graphs, given an oracle for queries on reachability from one node to another. As directed graphs are the basic model for many real-life networks including social, information, biological and technological networks (see~\cite{Newman03} for a survey), our algorithms provide a potential tool for the reconstruction of such networks.

There is a vast literature on algorithms for determining properties of an
initially unknown total order by means of comparisons. Partial orders often
arise in these studies as a representation of the ``information state'' at a
general step of such an algorithm. In such cases the incomparability of two
elements simply means that their true relation has not been determined yet.
The present work is quite different, in that the underlying structure to be
discovered is a partial order, and incomparability of elements is inherent,
rather than representing temporary lack of information. Nevertheless, the body
of work on comparison algorithms for total orders provides valuable tools and
insights that can be extended to the present context (e.g. \cite{Brightwell, Fredman, Kahn&Kim, Knuth, Linial}).

The model considered here was previously considered by
Faigle and Tur\'{a}n \cite{Faigle&Turan}, who presented
two algorithms for the problem of sorting a partial ordered set,
which they term ``identification'' of a poset.
We formally describe their results in Section~\ref{sec:sorting-kselection}.
A recent paper ~\cite{OnakParys} considers an
extension  of the searching and sorting problem to partial orders
that are either trees or forests.

\subsection{Definitions}

To precisely describe the problems considered in this paper and our
results, we require some formal definitions.  A partially ordered
set, or poset, is a pair $\po = (P, \dom)$, where $P$ is a set of
elements and $\dom\; \subset P \times P$ is an irreflexive,
transitive binary relation. For elements $a,b \in P$, if $(a,b)
\in\; \dom$, we write $a \dom b$ and we say that $a$
\emph{dominates} $b$, or that $b$ is \emph{smaller than} $a$. If
$a\not \dom b$ and $b \not \dom a$, we say that $a$ and $b$ are
\emph{incomparable} and write $a \incomp b$.

A \emph{chain} $C\subseteq P$ is a subset of mutually comparable
elements, that is, a subset such that for any elements $c_i, c_j \in
C$, $i\neq j$, either $c_i \dom c_j$ or $c_j \dom c_i$.
An \emph{ideal} $I\subseteq P$ is a subset of elements such that if $x \in I$ and $x \dom y$, then
$y \in I$.  The {\em height} of an element $a$ is the maximum cardinality of a chain
whose elements are all dominated by $a$.  We call the set
$\set{a\,:\,\forall\, b, \ b \dom a \text{ or } b \incomp a} $ of
elements of height 0 the \emph{minimal} elements. An
\emph{anti-chain} $A \subseteq P$ is a subset of mutually
incomparable elements. The {\it width} $\wid(\po)$ of poset $\po$ is
defined to be the maximum cardinality of an anti-chain of $\po$.

A \emph{decomposition} $\C$ of $\po$ into chains is a family
$\C=\set{C_1, C_2, \ldots, C_q}$ of disjoint chains such that their
union is $P$. The \emph{size} of a decomposition is the number of
chains in it. The width $\wid(\po)$ is clearly a lower bound on the
size of any decomposition of $\po$. We make frequent use of {\it
Dilworth's Theorem}, which states that there is a decomposition of
$\po$ of size $\wid(\po)$.  A decomposition of size $\wid(\po)$ is
called a \emph{minimum chain decomposition}.

\subsection{Sorting and $k$-selection}
\label{sec:sorting-kselection}
The central computational problems of this paper are \emph{sorting}
and \emph{$k$-selection}.  The sorting problem is to completely
determine the partial order on a set of $n$ elements, and the $k$-selection
problem is to determine the set of elements of height at most $k-1$, i.e., the
set of elements in the $k$ bottom levels of the partial order. In both problems
we are given an upper bound of $\wid$ on the width of the partial order.

In the absence of a bound on the width,  the  query complexity of the sorting
problem is exactly $n \choose 2$, in view of the worst-case example in which
all pairs of elements are incomparable.
In the classical sorting and selection problems,  $\wid = 1$.  Our interest is
mainly
in the case where
$\wid \ll n$,
since this assumption is natural in many of the
applications. Furthermore, if $\wid$ is of the same order as $n$,
then it is easy to see that the complexity of sorting is of order
$n^2$, as in the case where no restrictions are imposed on the
poset.

Faigle and Tur\'{a}n \cite{Faigle&Turan} have described two algorithms for sorting
posets, both of which have query complexity $O\left(\wid n \log \frac{n}{\wid}\right)$.
(In fact the second algorithm is shown to have query complexity
$O(n\log N_\po)$, where $N_\po$ is the number of ideals in input poset $\po$.
It is easy to see that $N_{\po} = O(n^{\wid})$ if $\po$ has width $\wid$, and that $N_{\po} = (n/\wid)^{\wid}$ if $\po$ consists of $\wid$ incomparable chains, each of size $n/\wid$.)  The total complexity
of sorting posets has not been considered.
However, the total complexity of the first algorithm given by Faigle and Tur\'{a}n depends
on the subroutine for computing a chain decomposition (the complexity of which is not analyzed in~\cite{Faigle&Turan}). It is not clear if there exists a polynomial-time implementation of the second algorithm.

\subsection{Techniques}

It is natural to approach the problems of sorting and $k$-selection
in posets by considering generalizations of the
well-known algorithms for the case of total orders,
whose running times are closely matched by proven lower bounds.
Somewhat surprisingly, natural generalizations of the classic algorithms
{\em do not} provide optimal poset algorithms
in terms of total and query complexity.

In the case of sorting, the
generalization of Mergesort considered here loses a
factor of $\wid$ in its total complexity compared to the
information-theoretic lower bound.
Interestingly, one can achieve the information-theoretic lower bound
on query complexity (up
to a constant factor) by carefully exploiting the structure of the
poset.  We do not know whether it is possible to achieve the
information-theoretic bound on total complexity.

The seemingly easier problem of $k$-selection still poses
some challenges. In particular, nontrivial arguments
are needed to obtain both lower and upper bounds. Moreover, there is a gap
of factor $2$ between the lower and upper bound, even for the problem of
finding minimal elements.

\subsection{Main Results and Paper Outline}

In Section~\ref{sec:chainmerge}, we briefly discuss
an efficient representation of a poset. The representation is of size
$O(\omega n)$, and it allows the relation between any two elements to be retrieved in time $O(1)$.

In Sections~\ref{sec:optimal} and~\ref{sec:efficient}, we prove the following main theorems:
\begin{theorem}
There exists an algorithm for sorting a poset of width at most $\wid$ over $n$
elements with optimal query complexity $O(n (log n + \wid))$.
\end{theorem}

\begin{theorem}
A generalization of Mergesort for sorting a poset of width at most $\wid$ over $n$
elements has query complexity $O(\wid n \log n)$ and total complexity
$O(\wid^2 n\log n)$. The algorithm also provides a minimum chain decomposition of
the set.
\end{theorem}

In Section~\ref{sec:selection},
we consider the $k$-selection problem of determining the elements of height
less than or
equal to $k-1$. We give upper and lower bounds on the query complexity and
total
complexity of $k$-selection within deterministic and randomized models of
computation.  For the case $k=1$ (finding the minimal elements),
we show
that the query complexity and total complexity are $\Theta(\wid n)$. The
query upper bounds match the query lower bounds up to a factor of $2$.

In Section~\ref{appx:karpsthing}, we give a randomized algorithm, based on a generalization of Quicksort, of
expected total complexity $O(n(\log n +\wid))$ for computing a linear extension of a poset.  We
also give a randomized algorithm of expected total complexity $O(\wid n \log n)$ for computing the heights of all elements in a poset.

Finally, in Section~\ref{sec:posetvariants}, we show that the results on sorting
posets generalize to the case when an upper bound on the width is not known
and to the case of transitive relations.

\subsection{Acknowledgments:}
E.M. would like to thank Mike Saks for the reference to the work of
Faigle and Tur\'{a}n~\cite{Faigle&Turan}.

\section{Representing a poset: the {\sc ChainMerge} data structure}
\label{sec:chainmerge}

Once the relation between every pair of elements in a poset has been
determined, some representation of this information is required,
both for output and for use in our algorithms.  The simple
$\chainmerge$ data structure that we describe here supports
constant-time look-ups of the relation between any pair of elements.
It is built from a chain decomposition of the poset.

Let $\C=\set{C_1, \ldots C_q}$ be a chain decomposition of a poset
$\po = (P,\dom)$.  $\chainmerge(\po, \C)$ stores, for each element
$x\in P$, $q$ indices as follows: Let $C_i$ be the chain of $\C$
containing $x$. The data structure stores the index of $x$ in $C_i$
and, for all $j$, $1 \leq j \leq q$, $j \neq i$, the index of the
largest element of chain $C_j$ that is dominated by $x$.
The performance of the data structure is
characterized by the following lemma.
\begin{claim}
\label{thm:chainmerge} Given a query oracle for a poset
$\po=(P,\dom)$ and a decomposition $\C$ of $\po$ into $q$ chains,
building the \chainmerge\ data structure has query complexity at
most $2qn$ and total complexity $O(qn)$, where $n=\card{P}$. Given
$\chainmerge(\po, \C)$, the relation in $\po$ of any pair of
elements can be found in constant time.
\end{claim}
\begin{proof}
The indices corresponding to chain $C_j$ that must be stored for the
elements in chain $C_i$ can be found in $O(\card{C_i}+\card{C_j})$
time, using $\card{C_i} + \card{C_j}$ queries, by simultaneously
scanning $C_i$ and $C_j$. Since each chain is scanned $2q-1$ times,
it follows that the query complexity of $\chainmerge(\po,\C)$ is at
most $2qn$, and the total complexity is $O(q \cdot \sum_{i=1}^{q}
\card{C_i})=O(qn)$.

Let $x,y\in P$, with $x\in C_i$ and $y\in C_j$.  The look-up
operation works as follows:  If $i=j$, we simply do a comparison on
the indices of $x$ and $y$ in $C_i$, as in the case of a total
order. If $i\neq j$, then we look up the index of the largest
element of $C_j$ that is dominated by $x$; this index is greater than
(or equal to) the index of $y$ in $C_j$ if and only if $x\dom y$. If
$x \not \dom y$, then we look up the index of the largest element of
$C_i$ that is dominated $y$; this index is greater than (or equal
to) the index of $x$ in $C_i$ if and only if $y \dom x$. If neither
$x\dom y$ nor $y\dom x$, then $x \incomp y$.
\end{proof}

\section{The sorting problem} \label{sec:sorting}

We address the problem of {\em sorting a poset}, which is the
computational task of producing a representation of a poset
$\po=(P,\dom)$, given the set $P$ of $n$ elements, an upper bound of $\wid$ on the
width of $\po$, and access to an oracle for $\po$. (See Section~\ref{sec:unknownwidth}
for a discussion of the case when an upper bound on the width is not known.)
An information-theoretic lower bound on the query complexity of sorting is
implied by the following theorem of Brightwell and Goodall
\cite{Brightwell&Goodall:SortingLowerBound}, which provides a lower bound
on the number $\nposets$ of posets of width at most $\wid$ on $n$
elements.

\begin{theorem}
\label{thm:info_sorting_lower_bound} The number $\nposets$
of partially ordered sets of $n$ elements and width at most $\wid$
satisfies
$$ \frac{n!}{\wid !}~4^{n(\wid-1)}~n^{-24\wid(\wid-1)} \le \nposets \le  n!~4^{n(\wid-1)}~n^{-(\wid-2)(\wid-1)/2} \wid^{\wid(\wid-1)/2}.$$
\end{theorem}

\noindent It follows that, for $\wid = o\left(\frac{n}{\log n}\right)$,
\begin{align*}
\log{\nposets} = \Theta(n \log n + \wid n). \label{eq:info_sorting_lower_bound}
\end{align*}

\subsection{An optimal sorting algorithm} \label{sec:optimal}
In this section, we describe a sorting algorithm that has optimal
query complexity, i.e. it sorts a poset of width at most $\wid$ on
$n$ elements using $\Theta(n \log n + \wid n)$ oracle queries.
Our algorithm is not necessarily computationally efficient, so in
Section~\ref{sec:efficient}, we consider efficient solutions to the
problem.

Before presenting our algorithm, it is worth discussing an intuitive approach that is different from the one we take.  For any set of oracle queries and responses, there is a corresponding set of posets, which we call \emph{candidates}, that are the posets consistent with the responses to these queries.
A natural sorting algorithm is to find a sequence of oracle queries such that, for each query (or for a positive fraction of the queries), the possible responses to it partition the space of posets that are candidates after the previous queries into three parts, at least two of which are relatively large.  Such an algorithm would achieve the information-theoretic lower bound (up to a constant).

For example, the effectiveness of Quicksort for sorting total orders
relies on the fact that most of the queries made by the algorithm partition the space of candidate total orders into two parts, each of relative size of at least $1/4$.
Indeed, in the case of total orders, much more is known:  for any subset of queries, there is a query that partitions the space of candidate total orders, i.e. linear extensions, into two parts, each of relative size of at least $3/11$ \cite{Kahn&Saks}.

In the case of width-$\wid$ posets, however, it could potentially be the case that most queries partition the space into three parts, one of which is much larger than the other two.
For example, if the set consists of $\wid$ incomparable chains, each of size $n/\wid$, then a random query has a response of incomparability with probability about $1-1/\wid$.  (On an intuitive level, this explains the extra factor of $\wid$ in the query complexity of our version of Mergesort, given in Section~\ref{sec:efficient}.) Hence, we resort to more elaborate sorting strategies.

Our optimal algorithm builds upon a basic algorithm that we call {\sc
Poset--BinInsertionSort}, which is identical to ``Algorithm A''
of Faigle and Tur\'{a}n~\cite{Faigle&Turan}.
The algorithm is inspired by the binary insertion-sort algorithm for total orders.  Pseudocode for {\sc Poset--BinInsertionSort}
is presented in Figure \ref{fig:poset-insertionsort-pseudo}.
The natural idea behind
{\sc Poset--BinInsertionSort} is to sequentially insert elements
into a subset of the poset, while maintaining a chain decomposition
of the latter into a number of chains equal to the width $\wid$ of
the poset to be constructed. A straightforward implementation of
this idea is to perform a binary search on every chain of the
decomposition in order to figure out the relationship of the element
being inserted with every element of that chain and, ultimately,
with all the elements of the current poset. It turns out
that this simple algorithm is not optimal; it is off by a factor of
$w$ from the optimum. In
the rest of this section, we show how to adapt {\sc
Poset--BinInsertionSort} to achieve the information-theoretic lower bound.

\begin{figure} \centering
\fbox{\begin{minipage}{6in}
\begin{program}
{\bf Algorithm} $\textsc{Poset--BinInsertionSort}(\po)$\\
~{\bf input:} a set $P$, a query oracle for a poset $\po=(P, \dom)$, and upper bound of $\wid$ on width of $\po$ \\
~{\bf output:} a $\chainmerge$ data structure for $\po$\\
\\
        \> {\bf 1.} ~$\po' := (\{e\},\{\})$, where $e \in P$ is some arbitrary element;~~~~/* $\po'$ is the current poset*/\\
        \> {\bf 2.} ~$P' := \{e\}$; ${\cal R}' := \{\}$;\\
        \> {\bf 3.} ~$U := P \setminus \{e\}$; ~~~~/* $U$ is the set of elements that have not been inserted */\\
        \> {\bf 4.} ~{\bf while} $U \neq \emptyset$\\
        \> \> ~~~{\bf a.} pick an arbitrary element $e \in U$; ~~~~/* $e$ is the element that will be inserted in $\po'$*/\\
        \> \> ~~~{\bf b.} $U := U \setminus \{e\}$;\\
        \> \> ~~~{\bf c.} find a chain decomposition $\C=\set{C_1,C_2, \ldots, C_q}$ of $\po'$, with $q \le \wid$ chains;\\
        \> \> ~~~{\bf d.} {\bf for} $i=1,\ldots, q$\\
        \> \> \> ~~~~~~{\bf ~~i.} let $C_i=\{e_{i1},\ldots,e_{i\ell_i}\}$, where $e_{i\ell_i} \dom \ldots \dom e_{i2} \dom e_{i1}$;\\
        \> \> \> ~~~~~~{\bf ~ii.} do binary search on $C_i$ to find smallest element (if any) that dominates $e$;\\
        \> \> \> ~~~~~~{\bf iii.} do binary search on $C_i$ to find largest element (if any) that is dominated by $e$;\\
        \> \> ~~~{\bf e.} based on results of binary searches, infer all relations of $e$ with elements of $P'$;\\
        \> \> ~~~{\bf f.} add into ${\cal R}'$ all the relations of $e$ with the elements of $P'$; $P':=P' \cup \{e\};$\\
        \> \> ~~~{\bf g.} $\po'=(P',{\cal R}')$;\\
        \> {\bf 5.} find a decomposition $\C$ of $\po'$; build $\chainmerge(\po', \C)$ (no additional queries needed);\\
        \> {\bf 6.} \RETURN $\chainmerge(\po', \C)$;
\end{program}
\caption{pseudo-code for {\sc Poset--BinInsertionSort}
\label{fig:poset-insertionsort-pseudo}}
\end{minipage}
}
\end{figure}

We begin by analyzing {\sc Poset--BinInsertionSort}.
\begin{lemma}[Faigle \& Tur\'{a}n~\cite{Faigle&Turan}]
{\sc Poset--BinInsertionSort} sorts any partial order $\po$ of width
at most $\wid$ on $n$ elements using at most $O(\wid n \log n)$ oracle queries.
\end{lemma}
\begin{proof}
The correctness of {\sc Poset--BinInsertionSort} should be clear
from its description. (The simple argument showing that Step 4e can be executed based on the information
obtained in Step 4d is similar to the proof for the $\chainmerge$ data structure in Section~\ref{sec:chainmerge}.)
It is not hard to see that the number of oracle queries incurred by {\sc Poset--BinInsertionSort}
for inserting each element is $O(\wid \log n)$ and, therefore, the
total number of queries is $O(\wid n \log n)$.
\end{proof}

It follows that, as $n$
scales, the number of queries incurred by the algorithm is more by a
factor of $\wid$ than the lower bound. The Achilles' heel of
the {\sc Poset-BinInsertionSort} algorithm is in the method of insertion of an element---specifically, in
the way the binary searches of Step 4d are performed.
In these sequences of queries, no structural properties of $\po'$ are used for deciding
which queries to the oracle are more useful than others; in some sense, the binary searches give the
same ``attention'' to queries whose answer would greatly decrease the
 number of remaining possibilities and those
whose answer is not very informative. However, as we
discuss earlier in this section, a sorting algorithm that always
makes the most informative query is not guaranteed to be optimal.

Our algorithm tries to resolve this dilemma. We suggest a
scheme that has the same structure as the {\sc
Poset--BinInsertionSort} algorithm but exploits the structure of
the already constructed poset $\po'$ in order to amortize the cost
of the queries over the insertions.  The amortized query cost
matches the information-theoretic bound.

The new algorithm, named {\sc EntropySort}, modifies the binary
searches of Step 4d into weighted binary searches. The weights
assigned to the elements satisfy the following property: the number
of queries it takes to insert an element into a chain is
proportional to the number of candidate posets that will be
eliminated after the insertion of the element. In other words, we spend fewer
queries for insertions that are not informative and more queries for insertions
that are informative. In some sense, this corresponds to an {\em
entropy-weighted binary search}.  To define this notion precisely, we use the following definition.

\begin{definition}
Suppose that $\po'=(P',{\cal R}')$ is a poset of width at most
$\wid$, $U$ a set of elements such that $U \cap P' = \emptyset$,
$u \in U$ and ${\cal ER}, {\cal PR} \subseteq (\{u\} \times P')
\cup (P'\times \{u\})$. We say that $\po = (P' \cup U, {\cal R})$
is {\em a width $\wid$ extension of $\po'$ on $U$ conditioned on
$({\cal ER},{\cal PR})$}, if $\po$ is a poset of width $w$, ${\cal
R} \cap (P'\times P') = {\cal R}'$ and, moreover, ${\cal ER}
\subseteq {\cal R}$, ${\cal R} \cap {\cal PR} = \emptyset$. In
other words, $\po$ is an extension of $\po'$ on the elements of
$U$ which is consistent with $\po'$, it contains the relations of $u$
to $P'$ given by $\cal ER$ and does not contain the relations of $u$ to $P'$ given by $\cal PR$.
The set $\cal ER$ is then called the {\em set of enforced
relations} and the set $\cal PR$ the {\em set of prohibited
relations}.
\end{definition}
\noindent We give in Figure~\ref{fig:poset-entropysortstep-pseudo}
the pseudocode of Step 4d$^\prime$ of {\sc EntropySort}, which replaces Step
4d of {\sc Poset--BinInsertionSort}.
\begin{figure} \centering
\fbox{\begin{minipage}{6in}
\begin{program}
{\bf Step 4d$^\prime$ for Algorithm} $\textsc{EntropySort}(\po)$\\
\\
{\bf 4d$^\prime$.}~ $\cal ER = \emptyset$; $\cal PR = \emptyset$;\\
\> {\bf for} $i=1,\ldots,q$\\
\>\> {\bf i.} let $C_i=\{e_{i1},\ldots,e_{i\ell_i}\}$, where
$e_{i\ell_i} \dom \ldots \dom e_{i2} \dom e_{i1}$;\\
\>\> {\bf ii.}~{\bf for} $j=1,\ldots,\ell_i+1$\\
\>\>\> $\bullet$ set ${\cal ER}_j= \{(e_{ik},e)|j \le k \le \ell_{i}\}$; set ${\cal PR}_j= \{(e_{ik},e)|1 \le k < j\}$;\\
\>\>\> $\bullet$ compute ${\cal D}_{ij}$, the number of $\wid$-extensions of
$\po'$ on $U$,\\
\>\>\>\> conditioned on $({\cal ER}\cup {\cal ER}_j, {\cal
PR} \cup {\cal PR}_j)$;\\
\>\>\>\> /* ${\cal D}_{ij}$ represents the number of
posets on $P$ consistent with $\po'$, $(\cal ER$, $\cal
PR$),\\
\>\>\>\>\> in which $e_{ij}$ is the smallest element of chain $C_i$
that dominates $e$;\\
\>\>\>\>\> $j=\ell_i+1$ corresponds to the case that no element of $C_i$ dominates $e$;*/\\
\>\>\> {\bf endfor}\\
\>\> {\bf iii.} set ${\cal D}_i = \sum_{j=1}^{\ell_i+1}{{\cal
D}_{ij}};$ \\
\>\>\> /* ${\cal D}_i$ is equal to the total number of $\wid$-extensions of $\po'$ on $U$\\
\>\>\>\> conditioned on $({\cal ER},{\cal PR})$*/\\
\>\> {\bf iv.} partition the unit interval
$[0,1)$ into $\ell_i+1$ intervals $([b_j,t_j))_{j=1}^{\ell_i+1}$,\\
\>\>\> where $b_1=0$, $b_j = t_{j-1}$, for all $j \ge 2$,\\
\>\>\> and $t_j = (\sum_{j' \leq j} {\cal D}_{ij'})/{\cal D}_i$, for all $j \ge 1$.\\
\>\>\> /* each interval corresponds to an element of $C_i$ or a ``dummy'' element $e_{i\ell_i+1}$ */ \\
\>\>{\bf v.} do binary search on $[0,1)$ to find smallest element (if any) of $C_i$ that dominates $e$: \\
\>\>\> /* weighted version of binary search in Step 4dii of {\sc Poset--BinInsertionSort} */ \\
\>\>\> set $x=1/2$; $t=1/4$; $j^*=0$; \\
\>\>\> {\bf repeat:} find $j$ such that $x \in [b_j,t_j)$; \\
\>\>\>\> {\bf if}~ ($j=\ell_i+1$ {\bf and} $e_{i,j-1} \ndom e$)~{\bf OR}~($e_{ij}\dom e$ {\bf and} $j=0$)~{\bf OR}~($e_{ij}\dom e$ {\bf and} $e_{i,j-1} \ndom e$) \\
\>\>\>\>\> {\bf set} $j^*=j$; break; /* found smallest element in
$C_i$ that dominates $e$ */ \\
\>\>\>\> {\bf else if} ($j=\ell_i+1$)~{\bf OR}~($e_{ij}\dom e$) \\
\>\>\>\>\> {\bf set} $x=x-t$; $t=t*1/2$; /* look below */ \\
\>\>\>\> {\bf else} \\
\>\>\>\>\> {\bf set} $x=x+t$; $t=t*1/2$; /* look above */ \\
\>\> {\bf vi.} $e_{ij^*}$ is the smallest element of chain $C_i$
that dominates $e$;\\ 
\>\>\>~ set ${\cal ER} := {\cal ER} \cup {\cal ER}_{j^*}$ and ${\cal PR} :=
{\cal PR} \cup {\cal PR}_{j^*}$;\\
\>\> {\bf vii.} find the largest element (if any) of
chain $C_i$ that is dominated by $e$:\\
\>\>\> for $j=0,1,\ldots,\ell_i$, \\
\>\>\>\> compute ${\cal D}'_{ij}$, the number of posets on $P$
consistent with $\po'$, $(\cal ER$, $\cal PR$),\\
\>\>\>\>~~~ in which $e_{ij}$ is the largest element of chain $C_i$ dominated by $e$;\\
\>\>\>\> /* $j=0$ corresponds to case that no element of $C_i$ is dominated by $e$; */\\
\>\>\> let ${\cal D}'_i = \sum_{j=0}^{\ell_i}{{\cal D}'_{ij}}$;\\
\>\>\> do the weighted binary search analogous to that of Step v;\\
\>\> {\bf viii.} update accordingly the sets ${\cal ER}$ and ${\cal PR}$;\\
\> {\bf endfor}
\end{program}
\caption{Algorithm {\sc EntropySort} is obtained by substituting Step 4d$^\prime$ above for Step 4d of the pseudo-code in Figure~\ref{fig:poset-insertionsort-pseudo} for {\sc Poset--BinInsertionSort}.
\label{fig:poset-entropysortstep-pseudo}}
\end{minipage}
}
\end{figure}

The correctness of the {\sc EntropySort} algorithm follows trivially
from the correctness of {\sc Poset--BinInsertionSort}. We prove
next that its query complexity is optimal. Recall that $\nposets$ denotes the number of
partial orders of width $w$ on $n$ elements.

\begin{theorem}\label{th:entropysortanalysis}
{\sc EntropySort} sorts any partial order $\po$ of width at most $\wid$
on $n$ elements using at most $2\log \nposets + 4\wid n = \Theta(n\log
n+\wid n)$ oracle queries. In particular, the query complexity of the algorithm is at most $2 n \log n + 8 \wid n + 2\wid \log \wid$.
\end{theorem}
\begin{proof}
We first characterize the number of oracle calls required by the
weighted binary searches.
\begin{lemma}[Weighted Binary Search]\label{lem:weighted_binary_search}
For every $j \in \{1,2,\ldots,\ell_i+1\}$, if $e_{ij}$ is the
smallest element of chain $C_i$ which dominates element $e$
($j=\ell_i+1$ corresponds to the case where no element of chain
$C_i$ dominates $e$), then $j$ is found after at most
$2\cdot(1+\log{\frac{{\cal D}_{i}}{{\cal D}_{ij}}})$ oracle
queries in Step {\em v.} of the algorithm described above.
\end{lemma}
\begin{proof}
Let $\lambda=\frac{{\cal D}_{ij}}{{\cal D}_{i}}$ be the length of
the interval that corresponds to $e_{ij}$. We wish to prove that
the number of queries needed to find $e_{ij}$ is at most $2
(1+\floor{\log \frac{1}{\lambda}})$. From the definition of the
weighted binary search, we see that if the interval corresponding to $e_{ij}$ contains a point of the form $2^{-r} \cdot
m$ in its interior, where $r,m$ are integers, then the search reaches $e_{ij}$
after at most $r$ steps. Now, an interval of length $\lambda$
must include a point of the form $2^{-r} \cdot m$, where
$r=1+\floor{\log \frac{1}{\lambda}}$, which concludes the proof.
\end{proof}

It is important to note that the number of queries spent by the
weighted binary search is small for uninformative insertions, which
correspond to large ${\cal D}_{ij}$'s, and large for informative
ones, which correspond to small ${\cal D}_{ij}$'s. Hence, our use of
the term entropy-weighted binary search. A parallel of Lemma
\ref{lem:weighted_binary_search} holds, of course, for finding the
largest element of chain $C_i$ dominated by element $e$.

Suppose now that $P=\{e_1,\ldots,e_n\}$, where $e_1,
e_2,\ldots,e_n$ is the order in which the elements of $P$ are
inserted into poset $\po'$. Also, denote by $\po_k$ the restriction
of poset $\po$ onto the set of elements $\{e_1, e_2, ..., e_k\}$
and by $Z_k$ the number of width $\wid$ extensions of poset
$\po_k$ on $P\setminus \{e_1,\ldots,e_k\}$ conditioned on
$(\emptyset, \emptyset)$. Clearly, $Z_0 \equiv \nposets$
and $Z_n =1$. The following lemma is sufficient to establish the
optimality of {\sc EntropySort}.

\begin{lemma} \label{lem:query complexity of entropy sort step}{\sc EntropySort} needs at most
$4\wid + 2\log{\frac{Z_k}{Z_{k+1}}}$ oracle queries to insert
element $e_{k+1}$ into poset $\po_{k}$ in order to obtain
$\po_{k+1}$.
\end{lemma}

\begin{proof} Let $\C=\{C_1,\ldots,C_q\}$ be the
chain decomposition of the poset $\po_k$ constructed at Step 4c of
{\sc EntropySort} in the iteration of the algorithm in which
element $e_{k+1}$ needs to be inserted into poset $\po_k$. Suppose
also that, for all $i \in \{1,\ldots,q\}$, $\pi_{i} \in
\{1,\ldots,\ell_i+1\}$ and $\kappa_i \in \{0,1,\ldots,\ell_i\}$
are the indices computed by the binary searches of Steps v. and
vii. of the algorithm. Also, let ${\cal D}_{i}$, ${\cal D}_{ij}$,
$j \in \{1,\ldots,\ell_i+1\}$, and ${\cal D}'_{i}$, ${\cal
D}'_{ij}$, $j \in \{0,\ldots,\ell_i\}$, be the quantities computed
at Steps ii., iii. and vii. It is not hard to see that the
following are satisfied
\begin{align*}
Z_k &= {\cal D}_1~~~~~~~~~~~~~~~~~~~~~~~~~{\cal D}'_{q\kappa_q} = Z_{k+1}\\
{\cal D}_{i\pi_i} &= {\cal D}'_i, \forall
i=1,\ldots,q~~~~~~~~{\cal D}'_{i\kappa_i}= {\cal D}_{i+1}, \forall
i=1,\ldots,q-1
\end{align*}
Now, using Lemma~\ref{lem:weighted_binary_search}, it follows that
the total number of queries required to construct $\po_{k+1}$ from
$\po_k$ is at most
$$\sum_{i=1}^q{\left(2+2 \log {\frac{{\cal D}_{i}}{{\cal D}_{i\pi_i}}} + 2 + 2\log {\frac{{\cal D}'_{i}}{{\cal D}'_{i\kappa_i}}}\right)} \le 4\wid + 2\log{\frac{Z_k}{Z_{k+1}}}.$$\end{proof}

\noindent Using Lemma~\ref{lem:query complexity of entropy sort step}, the query complexity of {\sc EntropySort} is
\begin{align*}\sum_{k=0}^{n-1}{(\text{\# queries needed to insert element $e_{k+1}$})} &= \sum_{k=0}^{n-1}{\left(4\wid + 2\log{\frac{Z_k}{Z_{k+1}}}\right)} \\&= 4\wid n + 2 \log{\frac{Z_0}{Z_n}} = 4\wid n+ 2\log \nposets.
\end{align*}
Taking the logarithm of the upper bound in Theorem~\ref{thm:info_sorting_lower_bound}, it follows that the number of queries required by the algorithm is
$2 n \log n + 8 \wid n + 2\wid \log \wid$.

\end{proof}

\subsection{An efficient sorting algorithm}
\label{sec:efficient}

In this section, we turn to the problem of efficient sorting.
Superficially, the \textsc{Poset-Mergesort} algorithm that we
present has a recursive structure that is similar to the classical
Mergesort algorithm.
The merge step is quite different, however; it makes crucial use of
the technical \peeling\ algorithm in order to efficiently maintain a
small chain decomposition of the poset throughout the recursion.
The \peeling\ algorithm, described formally in Section
\ref{subsec:peeling}, is a specialization of the classic flow-based bipartite-matching algorithm
\cite{FordFulkerson} that is efficient in the comparison model.

\subsubsection{Algorithm {\sc Poset-Mergesort}}
\label{subsec:poset-mergesort}

Given a set $P$, a query oracle for a poset $\po=(P, \dom)$, and an
upper bound of $\wid$ on the width of $\po$, the {\sc
Poset-Mergesort} algorithm produces a decomposition of $\po$ into
$\wid$ chains and concludes by building a \chainmerge\ data
structure. To get the chain decomposition, the algorithm partitions
the elements of $P$ arbitrarily into two subsets of (as close as
possible to) equal size; it then finds a chain decomposition of each
subset recursively. The recursive call returns a decomposition of
each subset into at most $\wid$ chains, which constitutes a
decomposition of the whole set $P$ into at most $2\wid$ chains. Then
the \peeling\ algorithm of Section~\ref{subsec:peeling} is applied
to reduce the decomposition to a decomposition of $\wid$ chains.
Given a decomposition of $P'\subseteq P$, where $m=\card{P'}$, into
at most $2\wid$ chains, the \peeling\ algorithm returns a
decomposition of $P'$ into $\wid$ chains using $2 \wid m$ queries and
$O(\wid^2 m)$ time. Figure~\ref{fig:poset-mergesort-pseudo}
shows pseudo-code
for \textsc{Poset-Mergesort}.

\begin{theorem} \label{thm:mergesort}
\textsc{Poset-Mergesort} sorts
any poset $\po$ of width at most $\wid$ on $n$ elements using at most $2\wid n\log (n/ \wid))$
queries, with total complexity $O(\wid^2 n\log (n/ \wid))$.
\end{theorem}

\begin{figure} \centering
\fbox{\begin{minipage}{6in}
\begin{program}
{\bf Algorithm} $\textsc{Poset-Mergesort}(\po)$ \\
~{\bf input:} a set $P$, a query oracle for a poset $\po=(P, \dom)$, and upper bound $\wid$ on width of $\po$ \\
~{\bf output:} a $\chainmerge$ data structure for $\po$ \\
\\
        \>     run $\textsc{Poset-Mergesort-Recurse}(P)$ producing a decomposition $\C$ of $\po$ into $\wid$ chains;\\
        \>     build and \RETURN $\chainmerge(\po, \C)$;\\
        \\

{\bf Procedure} $\textsc{Poset-Mergesort-Recurse}(P')$ \\
 ~{\bf input:} a subset $P' \subseteq P$, a query oracle for $\po = (P, \dom)$, an upper bound $\wid$ on the width of $\po$\\
 ~{\bf output:} a decomposition into at most $\wid$ chains of the poset $\po'$ induced by $\dom$ on $P'$ \\
\\
        \>     \IF $\card{P'} \le w$\\
        \>     \THEN \RETURN the trivial decomposition of
           $\po'$ into chains of length $1$\\
        \>     \ELSE \\
        \>     \> {\bf 1.}~    partition $P'$ into two parts of equal size, $P_1'$ and
$P_2'$;\\
        \>     \> {\bf 2.}~    run $\textsc{Poset-Mergesort-Recurse}(P_1')$
        and $\textsc{Poset-Mergesort-Recurse}(P_2')$; \\
        \>     \> {\bf 3.}~    collect the outputs to get a decomposition $\C$ of $\po'$ into $q \leq 2\wid$
chains;\\
        \>     \> {\bf 4.}~    if $q > \wid$, run $\peeling(\po, \C)$, to get a decomposition  $\C'$ of $\po'$ into $\wid$ chains;\\
        \>     \> \RETURN $\C'$;
\end{program}
\caption{pseudo-code for \textsc{Poset-Mergesort}
\label{fig:poset-mergesort-pseudo}}
\end{minipage}
}
\end{figure}

\begin{proof} The correctness of \textsc{Poset-Mergesort} is immediate. Let $T(m)$ and $Q(m)$
be the worst-case total and query complexity, respectively, of the
procedure \textsc{Poset-Mergesort-Recurse} on a poset of width
$\wid$ containing $m$ elements. When $m \le \wid$, $T(m)=O(\wid)$
and $Q(m)=0$. When $m> \wid$, $T(m)=2T(m/2)+O(\wid^2m)$ and $Q(m)
\le 2Q(m/2)+2\wid m$. Therefore, $T(n) = O(  \wid^2n \log (n/ \wid)
)$ and $Q(n) \le 2\wid n \log (n/ \wid)$. The cost incurred by the
last step of the algorithm, i.e. that of building the \chainmerge, is negligible.
\end{proof}

\subsubsection{The {\sc Peeling} algorithm} \label{subsec:peeling}

In this section we present an algorithm that efficiently reduces the
size of a given decomposition of a poset.
It can be seen as an adaptation of the classic flow-based bipartite-matching
algorithm \cite{FordFulkerson} that is designed to be efficient in the oracle model.
The \peeling\ algorithm is given an oracle for poset $\po = (P,
\dom)$, where $n=\card{P}$, and a decomposition of $P$ into at most
$2w$ chains. It first builds a $\chainmerge$ data structure using at
most $2qn$ queries and time $O(qn)$.  Every query the algorithm
makes after that is actually a look-up in the data structure and
therefore takes constant time and no oracle call.

The \peeling\ algorithm proceeds in a number of {\em peeling iterations}. Each iteration produces a decomposition of $\po$ with one less chain, until after at most $\wid$ peeling iterations, a
decomposition of $\po$ into $\wid$ chains is obtained.  A detailed formal description
of the algorithm is given in Figure~\ref{fig:peeling-pseudo}.

\begin{theorem} Given an oracle for $\po = (P, \dom)$, where $n=\card{P}$,
and a decomposition of $\po$ into at most $2\wid$ chains, the \peeling\ algorithm returns a decomposition of $\po$
into $\wid$ chains. It has query complexity at most $2\wid n$ and total complexity $O(\wid^2 n)$.
\end{theorem}

\begin{figure} \centering
\fbox{\begin{minipage}{6in}
\begin{program}
{\bf Algorithm} $\textsc{Peeling}(\po, \C)$\\
 {\bf input:} a query oracle for poset $\po = (P, \dom)$, an upper bound of $\wid$ on the width of $\po$, \\
  ~~~~~~~~~~~and a decomposition $C=\set{C_1, \ldots, C_q}$ of $\po$, where $q \leq 2w$ \\
 {\bf output:} a decomposition of $\po$ into $\wid$ chains\\
\\
\> build $\chainmerge(\po, \C)$; /* All further queries are look-ups. */\\
\> {\bf for}~ $i = 1, \ldots, q$\\
\>  \> construct a linked list for each chain $C_i = e_{i\ell_i}\rightarrow \cdots \rightarrow e_{i2} \rightarrow e_{i1}$,\\
\>\>\> where $e_{i\ell_i} \dom \cdots \dom e_{i2} \dom e_{i1}$;\\
\> {\bf while}~ $q > \wid$, perform a peeling iteration: \\
\>  \> {\bf 1.~ for}~ $i = 1, \ldots, q$, set $C'_i = C_i$;\\
\>  \> {\bf 2.~ while}~ every $C'_i$ is nonempty\\
\>\>  \>  \> /* the largest element of each $C'_i$ is a \emph{top element} */ \\
\>\>  \>  \> {\bf a.}~ find a pair $(x, y)$, $x\in C'_i$, $y\in C'_j$, of top elements such that $y\dom x$;\\
\>\>  \>  \> {\bf b.}~ delete $y$ from $C'_j$;  ~~ /* $x$ \emph{dislodges} $y$ */\\
\>  \> {\bf 3.~} in sequence of dislodgements, find subsequence $(x_1,y_1), \ldots, (x_t,y_t)$ such that:\\
\>\>  \>  \> $\bullet$~ $y_t$ is the element whose deletion (in step 2b) created an empty chain;\\
\>\>  \>  \> $\bullet$~ for $i =2,\ldots,t$, $y_{i-1}$ is the parent of $x_i$ in its original chain;\\
\>\>  \>  \> $\bullet$~ $x_1$ is the top element of one of the original chains;\\
\>  \> {\bf 4.~} modify the original chains $C_1,\ldots,C_q$:\\
\>\>  \> \> {\bf a.~ for}~  $i = 2,\ldots,t$\\
\>\>  \>  \>\> ~~~{\bf i.}~ delete the pointer going from $y_{i-1}$ to $x_i$;\\
\>\>  \>  \> \>~~~{\bf ii.}~ replace it with a pointer going from $y_i$ to $x_i$;\\
\>  \> \> \>{\bf b.~} add a pointer going from $y_1$ to $x_1$;\\
\>  \>  {\bf 5.~}set $q=q-1$, and re-index the modified original chains from 1 to $q-1$;\\
\> {\bf return}~ the current chain decomposition, containing $\wid$ chains
\end{program}
\caption{pseudo-code for the {\sc Peeling} Algorithm
\label{fig:peeling-pseudo}}
\end{minipage}
}
\end{figure}

\begin{proof}
To prove the correctness of one peeling iteration, we first observe
that it is always possible to find a pair $(x,y)$ of top elements
such that $y \dom x$, as specified in Step 1a, since the size of any
anti-chain is at most the width  of $\po$, which is less than the
number of chains in the decomposition.
We now argue that it is possible to find a subsequence of
dislodgements as specified by Step 2a. Let $y_t$ be the element
defined in step 3 of the algorithm. Since $y_t$ was dislodged by
$x_t$, $x_t$ was the top element of some list when that happened.
In order for $x_t$ to be a top element, it was either top from the
beginning, or its parent $y_{t-1}$ must have been dislodged by
some element $x_{t-1}$, and so on.

We claim that, given a decomposition into $q$ chains, one peeling iteration produces a
decomposition of $\po$ into $q-1$ chains.
Recall that $y_1 \dom
x_1$ and, moreover, for every $i$, $2\le i \le t$, $y_i \dom x_i$,
and $y_{i-1} \dom x_i$. Observe that after Step 4 of the
peeling iteration, the total number of pointers has increased by
$1$. Therefore, if the link structure remains a union of
disconnected chains, the number of chains must have decreased by
$1$, since $1$ extra pointer implies $1$ less chain. It can be seen
that the switches performed by Step 4 of the algorithm maintain the invariant
that the in-degree and out-degree of every vertex is bounded by $1$.
Moreover, no cycles are introduced since every pointer that is
added corresponds to a valid relation. Therefore, the link structure is indeed a union of
disconnected chains.

The query complexity of the \peeling\ algorithm is exactly the query
complexity of \chainmerge, which is $2\wid n$.  We show next that
one peeling iteration can be implemented in time $O(qn)$, which
implies the claim.

In order to implement one peeling iteration in time $O(qn)$, a
little book-keeping is needed, in particular, for Step 2a.  We
maintain during the peeling iteration a list $L$ of
potentially-comparable pairs of elements. At any time, if a pair
$(x,y)$ is in $L$, then $x$ and $y$ are top elements. At the
beginning of the iteration, $L$ consists of all pairs $(x,y)$ where
$x$ and $y$ are top elements. Any time an element $x$ that was not a
top element becomes a top element, we add to $L$ the set of all
pairs $(x,y)$ such that $y$ is currently a top element. Whenever a
top element $x$ is dislodged, we remove from $L$ all pairs that
contain $x$. When Step 2a requires us to find a pair of comparable
top elements, we take an arbitrary pair $(x,y)$ out of $L$ and check
if $x$ and $y$ are comparable. If they are not comparable, we remove
$(x,y)$ from $L$, and try the next pair. Thus, we never compare a
pair of top elements more than once. Since each element of $P$ is
responsible for inserting at most $q$ pairs to $L$ (when it becomes
a top element), it follows that a peeling iteration can be
implemented in time $O(qn)$.

\end{proof}

\section{The $k$-selection problem} \label{sec:selection}

The $k$-selection problem
is the natural problem of finding the elements in the bottom $k$
layers, i.e., the elements of height at most $k-1$, of a poset
$\po=(P, \dom)$, given the set $P$ of $n$ elements, an upper bound $\wid$ on the width of
$\po$, and a query oracle for $\po$. We present upper and lower
bounds on the query and total complexity of $k$-selection, both for
deterministic and randomized computational models, for the special
case of $k=1$ as well as the general version. While our upper bounds
arise from natural generalizations of analogous algorithms for total
orders, the lower bounds are achieved quite differently.  We
conjecture that our deterministic lower bound for the case of $k=1$
is actually tight, though the upper bound is off by a factor of 2.

\subsection{Upper bounds}
In this section we analyze some deterministic and randomized
algorithms for the $k$-selection problem.  We begin with the
$1$-selection problem, i.e., the problem of finding the minimal
elements.

\begin{theorem} \label{thm:ub_min_selection}The minimal elements
can be found deterministically with at most $\wid n$ queries and $O(\wid n)$ total complexity.
\end{theorem}

\begin{proof}
The algorithm updates a set of size $\wid$ of elements that are
candidates for being smallest elements. Initialize $T_0 =
\emptyset$. Assume that the elements are $x_1,\ldots,x_n$. At step
$t$:\\[-16pt]
\begin{itemize}
\item Compare $x_t$ to all elements in $T_{t-1}$.\\[-18pt]
\item 
If there exists some $a \in T_{t-1}$ such that $x_t \dom a$, do nothing.\\[-18pt]
\item
Otherwise, remove from $T_{t-1}$ all elements $a$ such that $a
\dom x_t$ and put $x_t$ into $T_t$.\\[-16pt]
\end{itemize}
At the termination of the algorithm, the set $T_n$
contains all height 0 elements. By construction of $T_t$, for all $t$, the elements in
$T_t$ are mutually incomparable. Therefore, for all $t$, it
holds that $|T_t| \leq \wid$, and hence the query complexity of
the algorithm is at most $\wid n$.
\end{proof}

\begin{theorem}
There exists a randomized algorithm that finds the minimal
elements in an expected number of queries that is upper bounded by
$$ \frac{\wid+1}{2} n + \frac{\wid^2-\wid}{2}(\log n-\log \wid).$$
\label{thm:randalg_min_selection}
\end{theorem}

\begin{proof}
The algorithm is similar to the algorithm for the proof of Theorem~\ref{thm:ub_min_selection},
with modifications to avoid (in expectation) worst-case behavior.
Let $\sigma$ be a permutation of $[n]$ chosen uniformly at random.
Let $T_1 = \set{x_{\sigma(1)}}$. For $1 \leq t < n$, at step $t$:\\[-14pt]
\begin{itemize}
\item Let $i$ be an index of the candidates in $T_{t-1}$, i.e. $T_{t-1} = \{ x_{i(1)},\ldots,x_{i(r)} \}$, where $r \leq \wid$. \\[-18pt]
\item Let $T_t = T_{t-1}$.  Let $\tau$ be a permutation of $[r]$ chosen uniformly at random. \\[-18pt]
\item For $j=1,\ldots,r$:\\[-18pt]
\begin{itemize}
\item If $x_{\sigma(t)} \dom x_{i(\tau(j))}$, exit the loop and
move to step $t+1$.\\[-14pt]
\item If $x_{i(\tau(j))}\dom x_{\sigma(t)}$,
remove $x_{i(\tau(j))}$ from $T_t$.\\[-14pt]
\end{itemize}
\item Add $x_{\sigma(t)}$ to $T_t$.\\[-12pt]
\end{itemize}

\noindent As in the previous algorithm, it is easy to see that at
each step $t$, the set $T_t$ contains all the minimal elements of
$A_t = \{ x_{\sigma(1)},\ldots,x_{\sigma(t)}\}$ and that $|T_t| \leq
\wid$. Note furthermore that at step $t$,
\[\pr[x_{\sigma(t)} \mbox{ is minimal for } A_t] \leq
\frac{\wid}{t}.
\]
If $x_{\sigma(t)}$ is not minimal for $A_t$, then the expected
number of comparisons needed until $x_{\sigma(t)}$ is compared to an
element $a\in A_t$ that dominates $x_{\sigma(t)}$ is clearly at most
$(\wid+1)/2$. We thus conclude that the expected running time of the
algorithm is bounded by:
\begin{eqnarray*}
\sum_{t=2}^{\wid} (t-1) + \sum_{t=\wid+1}^n \left( \frac{\wid}{t}
\wid + \frac{(t-\wid)}{t} \frac{(\wid+1)}{2} \right) &=&
\binom{\wid}{2} +  \sum_{t=\wid+1}^n \frac{1}{2t} \left(\wid^2 -
  \wid + t \wid + t \right)\\ 
  &\leq&
\frac{\wid+1}{2} n + \frac{\wid^2-\wid}{2}(\log n-\log \wid)
\end{eqnarray*}
\end{proof}

We now turn to the $k$-selection problem for $k > 1$. We first
provide deterministic upper bounds on query and total complexity.
\begin{theorem}
The query complexity of the $k$-selection problem is at most $$16 \wid n + 4 n \log{(2 k)}  + 6 n \log{\wid}.$$  Moreover, there exists an efficient $k$-selection
algorithm with query complexity at most $$8\wid n\log{(2k)}$$ 
and total complexity $$O(\wid^2 n\log (2k)).$$
\end{theorem}
\begin{proof}
The basic idea is to use the sorting algorithm presented in previous
sections in order to update a set of candidates for the
$k$-selection problem. Denote the elements by $x_1,\ldots,x_n$. Let
$C_0 = \emptyset$. The algorithm proceeds as follows, beginning with
$t$=1:\\[-14pt]
\begin{itemize}
\item While $(t-1) \wid k + 1\leq n$, let $ D_t =
C_{t-1} \cup \{ x_{(t-1) \wid k + 1},\ldots,x_{\min(t \wid k,n)} \}.
$\\[-18pt]
\item Sort $D_t$. Let $C_t$ be the solution of the
$k$-selection problem for $D_t$.\\[-14pt]
\end{itemize}

\noindent Clearly, at the end of the execution, the last $C_t$
will contain the solution to the $k$-selection problem. As we have
shown, the query complexity of sorting $D_t$ is $4 \wid k \log {(2 \wid k)} + 16 \wid^2 k + 2\wid \log \wid$ and, therefore,
the query complexity of the algorithm is $\frac{n}{\wid k}
(4 \wid k \log {(2 \wid k)} + 16 \wid^2 k + 2\wid \log \wid) = 4 n \log{(2 \wid k)} + 16 \wid n + \frac{2n}{k} \log{\wid}.$ This proves
the first result. Using the computationally efficient sorting
algorithm, we have sorting query complexity $8\wid^2k\log{(2k)}$ 
which results in total query complexity $8n\wid \log{(2k)}$ 
and total complexity $O(n \wid^2 \log(2k)).$
\end{proof}

\bigskip \noindent Next we
outline a randomized algorithm with a
better coefficient of the main term $\wid n$.

\begin{theorem} \label{thm:randalg1_k_selection}
The $k$-selection problem has a randomized query complexity of at most
$$\wid n + 16 k \wid^2 \log n \log (2k)$$
and total complexity
$$O(\wid n + poly(k,\wid) \log n).$$
\end{theorem}
\begin{proof}
We use the following algorithm:\\[-16pt]
\begin{itemize}
\item Choose an ordering $x_1,\ldots,x_n$ of the elements uniformly at
random.\\[-20pt]
\item Let $C_{\wid k} = \{ x_1, \ldots, x_{\wid k} \}$ and $D_{\wid k} =
\emptyset$.\\[-20pt]
\item Sort $C_{\wid k}$.
Remove any elements from $C_{\wid k}$ that are of height greater
than $k-1$.\\[-20pt]
\item Let $t=\wid k +1$. While $t \leq n$ do:\\[-18pt]
\begin{itemize}
\item Let $C_t = C_{t-1}$ and $D_t = D_{t-1}$. \item Compare $x_t$
to the maximal elements in $C_t$ in a random order.\\[-12pt]
\begin{itemize}
\item For each maximal element $a\in C_t$:  if $\height(a) = k-1$ and
$a \dom x_t$, or if $\height(a) < k-1$ and $x_t \dom a$, then add
$x_t$ to $D_t$, and exit this loop.\\[-14pt]
\item If for all elements $a
\in C_t$, $x_t \incomp a$, then add $x_t$ to $D_t$ and exit this loop;\\[-12pt] 
\end{itemize}
\item If $|D_t| = \wid k$ or $t=n$:\\[-12pt]
\begin{itemize}
\item Sort $C_t \cup D_t$.\\[-14pt]
\item Set $C_t$ to be the elements of
height at most $k-1$ in $C_t \cup D_t$.\\[-14pt]
\item Set $D_t = \emptyset$.\\[-18pt]
\end{itemize}
\end{itemize}
\item Output the elements of $C_n$.\\[-14pt]
\end{itemize}

\noindent It is clear that $C_n$ contains the solution to the
$k$-selection problem.  To analyze the query complexity of the
algorithm, recall from Theorem \ref{thm:mergesort} that $s(\wid,k) = 8\wid^2 k\log (2k)$ is an upper bound
on the number of queries used by the efficient sorting algorithm to sort $2 \wid k$ elements in a width-$\wid$ poset.

There are two types of contributions to the number of queries made by the algorithm:
(1) comparing elements to the maximal elements of $C_t$, and (2) sorting the sets $C_0$ and $C_t \cup D_t$.

To bound the expected number of queries of the first type, we note
that for $t \geq k \wid+1$, since the elements are in a random
order, the probability that $x_t$ ends up in $D_t$ is at most
$\min\left(1, \frac{2k\wid}{t}\right)$.
If $x_t$ is not going to be in $D_{t}$, then
the number of queries needed to verify this is bounded by $\wid$.
Overall, the expected number of queries needed for comparisons to
maximal elements is bounded by $\wid n$.

To calculate the expected number of queries of the second type,
we bound the expected number of elements that
need to be sorted as follows:
$$\sum_{t=k\wid+1}^n \min\left(1,\frac{2k\wid}{t}\right) \leq 2k \wid (\log n-1).
$$

\noindent We thus obtain that the
total query complexity is bounded above by $\wid n + 2s(\wid,k)
\log n.$
\end{proof}

\subsection{Lower bounds}
\label{subsec:selection_lower_bounds} We obtain lower bounds for the
$k$-selection problem both for adaptive and non-adaptive
adversaries. Some of our proofs use the following lower bound on
finding the $k$-th smallest element of a total order on $n$
elements:

\begin{theorem}[Fussenegger-Gabow~\cite{Fuss&Ga}]\label{thm:lb-totalorder}
The number of queries required to find the $k$th smallest element
of an $n$-element total order is at least $n - k +
\log\binom{n}{k-1}.$
\end{theorem}

\noindent The proof of Theorem~\ref{thm:lb-totalorder} shows that every comparison tree that
identifies the $k$th smallest element must have at least
$2^{n-k}\binom{n}{k-1}$ leaves, which implies that the theorem
also holds for randomized algorithms.

\subsubsection{Adversarial lower bounds}

We consider adversarial lower bounds for the $k$-selection problem.
In this model, an adversary simulates the oracle and is allowed to
choose her response to a query after she receives it. Any response
is legal as long as there is some partial order of width $\wid$ with
which all of her responses are consistent.  We begin with the case
of $k=1$, i.e. finding the set of minimal elements.

\begin{theorem} In the adversarial model, at least
$\frac{\wid + 1}{2} n - \wid$ comparisons are needed in order to
find the minimal elements.
\label{thm:adv_min_selection}
\end{theorem}

\noindent\begin{proof}  Consider the following
adversarial algorithm. The algorithm outputs query responses that
correspond to a poset $\po$ of $\wid$ disjoint chains. Given a
query $q(a,b)$, the algorithm outputs a response to the query, and
in some cases, it may also announce for one or both of $a$ and
$b$ to which chain the element belongs.  Note that receiving this
extra information can only make things easier for the query
algorithm.
During the course of the algorithm, the adversary
maintains a graph $G=(P,E)$.
 Whenever the adversary responds that $a \incomp b$, it adds an edge $(a,b)$ to $E$.

Let $q_t(a)$ be the number of queries that involve element $a$, out
of the first $t$ queries overall.  Let $c(a)$ be the chain
assignment that the adversary has announced for element $a$.  (We
set $c(a)$ to be undefined for all $a$, initially.) Let
$\set{x_i}_{i=1}^n$ be an indexing, chosen by the adversary, of the
elements of $P$.  Let $q(a,b)$ be the $t$'th query.  The adversary
follows the following protocol:\\[-14pt]
\begin{itemize}
\item If $q_t(a) \leq \wid-1$ or $q_t(b) \leq \wid-1$, return $a
\incomp b$. In addition:\\[-14pt]
\begin{itemize}
\item If $q_t(a) = \wid-1$, choose a chain $c(a)$ for $a$ that is
different from all the chains to which $a$'s neighbors in $G$
belong, and output it.\\[-14pt]
\item If $q_t(b) = \wid-1$ choose a chain
$c(b)$ for $b$ that is different from all the chains to which
$b$'s neighbors in $G$ belong, and output it.\\[-14pt]
\end{itemize}
\item If $q_t(a) > \wid-1$, $q_t(b) > \wid-1$, and
$c(a) \neq c(b)$, then output $a \incomp b$.\\[-14pt]
\item Otherwise, let $i$ and $j$
be the indices of $a$ and $b$, respectively (i.e. $a=x_i$ and
$b=x_j$).  If $i > j$, then output $a \dom b$; otherwise, output $b
\dom a$.\\[-14pt]
\end{itemize}
It is easy to see that the output of the algorithm is consistent
with a width-$\wid$ poset consisting of $\wid$ chains that are
pairwise incomparable. We will also require that each of the chains is chosen
at least once (this is easily achieved).

We now prove a lower bound on the number of queries to this
algorithm required to find a proof that the minimal elements are
indeed the minimal elements.

In any proof that $a$ is \emph{not} a smallest element, it must be shown
to dominate at least one other element, but to get such a response
from the adversary, $a$ must be queried against at least $\wid-1$
other elements with which it is incomparable.
To prove that a minimal element of one chain is indeed minimal, it
must be queried at least against the minimal elements of the other chains
to rule out the possibility it dominates one of them.  Therefore, each element must
be compared to at least $\wid-1$ elements that are incomparable to
it. So the total number of queries of type $q(a,b)$, where $a \incomp b$, is at
least $ \frac{\wid-1}{2} n$.

In addition, for each chain $c_i$ of length $n_i$, the output must
provide a proof of minimality for the minimal element of that
chain. By Theorem~\ref{thm:lb-totalorder}, this contributes $n_i - 1$ comparisons for each
chain $c_i$.

Summing over all the bounds proves the claim.
\end{proof}

\begin{theorem}\label{thm:adv_k_selection}
Let $r = \frac{n}{2w - 1}$. If $k \leq r$ then the number of queries
required to solve the $k$-selection problem is at least
\begin{align*}
\frac{(w+1)n}{2} - w(k+\log k) -\frac{w^3}{8}+
\min\bigg(\begin{aligned}[t]&(w-1) \log{r \choose k-1}+ \log {rw \choose k-1},\\
& \frac{n(r-k)(w-1)}{2r} +
\log {n - (w-1)k \choose k-1}\bigg).
\end{aligned}
\end{align*}
\end{theorem}

\begin{proof}
The adversarial algorithm outputs query responses exactly as in the
proof of Theorem \ref{thm:adv_min_selection}, except in the case where the
$t$th query is $(a,b)$
and $q_t(a) = w-1$ or $q_t(b) = w-1$. In that case it uses a more specific
rule for the assignment of one or both of these elements to chains.

In addition to assigning the elements to chains, the process must also select
the $k$ smallest elements in each chain, and the Fussenegger-Gabow theorem
(Theorem \ref{thm:lb-totalorder})
gives a lower bound, in terms of the lengths of the chains, on the number of
comparisons required to do so.

We think of the assignment of elements to chains as a coloring of the elements
with $w$ colors.
 The specific color  assignment rule is designed to ensure
that, if the number of elements with color $c$ is small, then there
must have been many queries in which  the element being colored  could not
receive color $c$  because it had already been declared incomparable to an
element with color $c$. It will then follow that there have been a large
number of queries in which an element was declared incomparable to an element
with color $c$.  Thus, if many of the chains are very short, then
the number of pairs declared incomparable must be very large. On the other
hand, if few of the chains are very short, then we can employ the
Fussenegger-Gabow Theorem to show that the number of comparisons required to
select the $k$ smallest elements in each chain must be large. We obtain the
overall lower bound by playing off these two observations against each other.

The color assignment rule is based on a function $d_t(c)$, referred to as
the {\it deviation}
of color $c$ after query $t$, and satisfying the initial condition $d_0(c) = 0$
for all $c$. The rule is: ``assign the eligible color with
smallest deviation.''

More specifically, let the $t$th query be $(a_t,b_t)$.  The adversary
processes $a_t$ and then $b_t$.  Recall that $q_t(a)$ is the number
of queries involving element $a$ out of the first $t$ queries overall.
Element $e \in \{a_t,b_t\}$ is processed exactly as in the
proof of Theorem \ref{thm:adv_min_selection} except when $q_t(e) = w-1$. In that case, let $S_t(e)$ be
the set of
colors that are {\it not} currently assigned to neighbors of $e$; i.e., the
set of colors eligible to be assigned to element $e$. Let
$c* = \argmin_{c \in S_t(e)} d_{t-1}(c)$. The adversary assigns  color
$c*$ to $e$. Then the deviations of all colors are updated as follows:
\begin{enumerate}
\item if $c \not \in S_t(e)$ then $d_t(c) = d_{t-1}(c)$;
\item $d_t(c^*) \leftarrow d_{t-1}(c^*) + 1 - \frac{1}{|S_t(e)|}$;
\item For $c \in S_t(e) \setminus \{c^*\}$,
$d_t(c) \leftarrow d_{t-1}(c)-\frac{1}{\card{S_t(e)}}$.
\end{enumerate}

The function $d_t(c)$ has the following interpretation: over the history of the
color assignment process, certain steps occur where the adversary has the
choice of whether to assign color $c$ to some element; $d_t(c)$ represents
the number of times that color $c$ was chosen up to step $t$, minus the
expected number of times it would have been chosen if the same choices had
been available at all steps and the color had been chosen uniformly at random
from the set of eligible colors.

Because the smallest of the deviations of eligible colors is
augmented at each step,
it is not possible for any deviation to drift far from zero. Specifically,
it can be shown by induction on $t$ that at every step $t$ the sum of the
deviations is zero
and for $m = 1,2,\ldots,w$, the sum of the $m$ smallest deviations
is greater than or equal to $\frac{m(m-w)}{2}$.

Let $\deg_G(a)$ be the degree of $a$ in $G$ at the end of the process.
At the end of the process every element of degree greater than or equal to
$w-1$ in $G$ has been assigned to a chain. Each element of degree less than
$w-1$ has not been assigned to a chain, and is therefore called
{\it unassigned}. An unassigned element is called {\it eligible} for chain $c$
if it has not been compared (and found incomparable) with any element of
chain $c$. Let $s(c)$ be the length of chain $c$ and define $\defic(c)$, the
{\em deficiency} of chain $c$, as $\max(0,k-s(c))$. Define the
{\em total deficiency} $\DEF$ as the sum of the deficiencies of all chains.

Let $u$ be the number of unassigned elements.
Upon the termination of the process it must be possible to infer from the
results of the queries that every unassigned
element is of height at most $k-1$. This implies that, if unassigned element
$x$ is eligible for chain $c$, then the number of unassigned elements eligible
for chain $c$ must be at most $\defic(c)$. Thus the number of pairs $(a,c)$ such
that unassigned element $a$ is eligible for chain $c$ is $\DEF$. Define the
{\it deficiency} of unassigned element $a$  as $w - 1 - \deg_G(a)$.
Then  the sum of the deficiencies of the unassigned elements is bounded
above by $\DEF$, and therefore the sum of the degrees in $G$ of the
unassigned elements is at least $(w - 1)u - \DEF$.

By Theorem \ref{thm:lb-totalorder}, if $s(c) > k$, then
at least $\left(s(c) -k + \log{s(c) \choose k-1}\right)$
comparisons are needed to determine the $k$ smallest elements of chain $c$.

The total number of comparisons is the number of edges that have been placed
in $G$ in the course of the algorithm (i.e., the number of pairs that have
been declared incomparable by the adversary), plus the number of comparisons
required to perform $k$-selection in each chain.
The total number of pairs that have been declared incomparable is
$\frac{1}{2}\sum_a \deg_G(a)$.

Let $d(c)$ be the
deviation of color $c$ at the end of the process. Let $r(c)$ be the
number of steps in the course of the process at which the element being
colored was eligible to
receive color $c$.  If, at each such step, the color had been chosen uniformly
from the set of eligible colors, then the chance of choosing color $c$ would
have been at least $\frac{1}{w}$.
Thus, by the interpretation of the
function $d_t(c)$ given above, $s(c) \geq \frac{r(c)}{w} + d(c)$;
equivalently,
$r(c) \leq w(s(c) - d(c))$.  Also,
$\sum_{a|c(a) = c} \deg_G(a) \geq  n - r(c) \geq n -w(s(c) - d(c))$.
This sum is also at least $(w-1)s(c)$, since every element assigned to $c$
has been  declared
incomparable with at least $(w-1)$ other elements.

We can now combine these observations to obtain our lower bound.
For each chain $c$ define
$\cost(c) = \frac{1}{2} \sum_{a|c(a) = c} \deg_G(a) +
\max\left(0, s(c) - k + \log{s(c) \choose k-1}\right)$. Then
$\sum_c \cost(c) + \frac{1}{2}\left((w-1)u - \DEF\right) $  is a lower
bound on the total number of comparisons, and
\begin{align*}
\sum_c \cost(c) \geq \begin{aligned}[t]&
\frac {1}{2} \sum_c \max \left((w-1)s(c), n - w(s(c) - d(c))\right)\\
& + \sum_{c|s(c) > k} \left( s(c) - k +
\log{s(c) \choose k-1}\right).\end{aligned}
\end{align*}

To obtain our lower bound we shall minimize this function over all choices of
nonnegative integers  $s(c)$, $u$ and $\DEF$  such that $\sum_c s(c) + u = n$
and $\DEF = \sum_c \max(0, k - s(c))$. Noting that
$\sum_c \min(d(c), 0) \geq \min_m m(m-w)/2 = -w^2/8$, we obtain the following
lower bound on the total number of comparisons:
\begin{equation}
\frac{(w-1)n}{2} -\frac{\DEF}{2} - \frac{w^3}{8} + \frac{1}{2}\sum_c \max
\left(0, n - (2w -1)s(c)\right) + \sum_{c|s(c) > k} \left(s(c) - k + \log{s(c)
\choose k-1}\right)  \label{eq:minobj}
\end{equation}

We now restrict attention to the case $k \leq \frac{n}{2w-1}$.
Let $r=\frac{n}{2 \wid -1}$. We shall show that, at any global minimum
of \eqref{eq:minobj}, $\DEF = 0$. To see this, consider any
choice of $\{s(c)\}$ such that $\DEF > 0$.
Let $c$ be a chain such that $\defic(c) > 0$.
If $s(c)$ is increased by 1, then $\DEF$ decreases by 1, and the net change in the value of quantity \eqref{eq:minobj} is $1-\wid$, which is negative.

Thus, in minimizing \eqref{eq:minobj} we may assume that $\DEF = 0$, and hence
that
$\sum_c s(c) = n$.  So \eqref{eq:minobj} may be rewritten as
$$ \frac{(w-1)n}{2} - \frac{w^3}{8} + \sum_c F(s(c))$$

where
\begin{equation*}
F(s) = \begin{cases}
\frac{1}{2} \max(0, n-(2w-1)s) &\text{ if }
s \leq k\\
\frac{1}{2} \max (0, n - (2w -1)s)
+ \left(s - k + \log{s \choose k-1}\right) &\text{ if } s > k.
\end{cases}
\end{equation*}
Thus, we have the following minimization problem:
\begin{align*}
\text{Minimize } \sum_c F(s(c)), \text{ subject to } s(c) \geq 0 \text{ and }\sum_c s(c) =n.
\end{align*}

First, we note that $\sum_{c|k<s(c)}(s(c)-k) =n-wk$.
To determine the minimum we consider three ranges of values:  the low range
$s = k$, medium range $k < s(c) \leq r$, and
high range $r <s(c) \leq n$. Observing that $F(s)$  is strictly
concave in the medium range, and concave and
strictly increasing in the high range, it follows  that,
at the global minimum of \eqref{eq:minobj}, $s(c)$ is equal to
either $k$ or $r$ except for one value in the high range and possibly
one value  strictly within the medium range. The value in the high
range is at least $rw$,
since the sum of the values in the low
and medium ranges does not exceed $r(w-1)$. If $\sum_{c|k \leq s(c) \leq r} s(c)= (w-1)r - D$, then the unique
value of $s(c)$ in the high range is $rw+D$.  Moreover, exploiting the
concavity of $F(s)$ in the medium range, we claim that
$\sum_{c|k \leq s(c) \leq r} \log{s(c) \choose k-1}
\geq (w - 1 - \frac{D}{r-k}) \log {r \choose k-1}$.  This bound is at most $w\log k$ greater
than the sum $\sum_{c|k<s(c)\leq r}\log{s(c) \choose k-1}$.
Finally, a simple calculation shows that $\frac{1}{2}\sum_{c}\max(0, n- (2w-1)s(c)) = \frac{nD}{2r}$.

Thus we get the following lower bound on $\sum_c F(s(c))$:\\
$$n-w(k+\log k)+\min_{0 \leq D \leq (w-1)(r-k)}\left(
\left(w - 1 - \frac{D}{r-k}\right) \log {r \choose k-1} + \frac{nD}{2r}
+ \log{rw +D \choose k-1}\right).$$

 Since this is a concave function it is
minimized either at $D=0$ or $D = (w-1)(r-k)$.This yields the following lower
bound  on the worst-case number of comparisons required
to solve the $k$-selection problem when $k \leq r$:\\
$\frac{(w+1)n}{2} - w(k+\log k) -\frac{w^3}{8}+
\min\left((w-1) \log{r \choose k-1}+ \log {rw \choose k-1},
\frac{n(r-k)(w-1)}{2r} +
\log {n - (w-1)k \choose k-1}\right)$.
\end{proof}

\subsubsection{Lower bounds in the randomized query model} We now
prove lower bounds on the number of queries used by randomized $k$-selection algorithms.
We conjecture that the randomized
algorithm for finding the minimal elements which we give in the proof of Theorem~\ref{thm:randalg_min_selection}
essentially achieves the lower bound. However, the lower bound we prove here is a factor $2$
different from this upper bound.

We consider a distribution $D(n,\wid)$ on partial orders of width
$\wid$ over a set $P = \set{x_1,\ldots,x_n}$.  The distribution
$D(n,\wid)$ is defined as follows:\\[-16pt]
\begin{itemize}
\item The support of $D(n,\wid)$ is the set of partial orders
consisting of $\wid$ chains, where any two elements from different
chains are incomparable.\\[-16pt]
\item Each element belongs independently to one of the $\wid$ chains with equal
probability. \\[-16pt]
\item The linear order on each chain is chosen uniformly.\\[-8pt]
\end{itemize}

\begin{theorem} \label{thm:random_k_selection} The
expected query complexity of any algorithm solving the
$k$-selection problem is at least $$ \frac{\wid+3}{4} n -\wid k +
\wid \left( 1 - \exp\left(-\frac{n}{8\wid}\right)\right)
\left(\log \binom{n/(2\wid)}{k-1} \right). $$
\end{theorem}
\begin{proof}
In order to provide a lower bound on the number of queries, we
provide a lower bound on the number of queries of
incomparable elements and then use the classical bound to bound
the number of queries of comparable elements.

First we note that for each element $a$, the algorithm must make
either at least one query where $a$ is comparable to some other
element $b$, or at least $\wid-1$ queries where $a$ is
incomparable to all elements queried.  (The latter may suffice in
cases where $a$ is the unique element of a chain and it is
compared to all minimal elements of all other chains.)

We let $Y_t(i)$ denote the number of queries involving $x_i$ {\em
before} the first query for which the response is that $x_i$ is
comparable to an element. Also for each of the chains
$C_1,\ldots,C_{\wid}$ we denote by $Z_{\alpha}$ the number of
comparisons involving two elements from the same chain.

Letting $T$ denote the total number of queries before the
algorithm terminates, we obtain:

$$\E[T] \geq \sum_{i=1}^n \frac{1}{2}\E(Y_T(i)) +
\sum_{\alpha=1}^{\wid} \E[Z_{\alpha}].$$

\smallskip We claim that for all $1 \leq i \leq n$ we have $
\E[Y_T(i)] \geq \frac{\wid-1}{2}.$ This follows by conditioning on
the chains that all other elements but $x_i$ belong to. With
probability $1/\wid$, the first query will give a comparison; with
probability $1/\wid$, the second query, etc.

On the other hand, by the classical lower bounds we have for each
$1 \leq \alpha \leq \wid$ that\\
$$ Z_{\alpha} \geq |C_{\alpha}| - k + \log
\binom{|C_{\alpha}|}{k-1} $$

\medskip\noindent Taking expected value we obtain
$$ \E[Z_{\alpha}] \geq \frac{n}{\wid} - k + \E\left[\log \binom{|C_{\alpha}|}{k-1}\right].$$
A rough bound on the previous expression may be
obtained by using the fact that by standard Chernoff bounds,
except with probability $\exp(-\frac{n}{8\wid})$, it holds that
$C_{\alpha}$ is
  of size at least $n / (2 \wid)$.
Therefore
$$\E\left[\log \binom{|C_{\alpha}|}{k-1}
\right] \geq \left(1-\exp\left(-\frac{n}{8\wid}\right)\right) \log
\binom{n/(2\wid)}{k-1}.$$

\medskip\noindent Summing all of the expressions above, we
obtain
$$\frac{(\wid-1)n}{4} + \wid
\left(\frac{n}{\wid} - k\right) + \left(1 -
\exp\left(-\frac{n}{8\wid}\right)\right) \wid \log
\binom{n/(2\wid)}{k-1}$$
and simplifying gives the desired result.
\end{proof}

\section{Computing linear extensions and heights}\label{appx:karpsthing}
In this section we
consider two problems that are closely related to the problem of
determining a partial order:  given a poset, compute a linear extension, and compute
the heights of all elements.

A total order $(P, >)$ is a {\it linear extension} of a
partial order $(P, \dom)$ if, for any two elements $x$ and $y$,
$x \dom y$ implies $x > y$.  We give a
randomized algorithm that, given a set $P$ of $n$ elements and access to an
oracle for a poset $(P, \dom)$ of width at most $\wid$, computes a linear extension of $(P, \dom)$
with expected total complexity $O(n(\log n + w))$.  We give another
randomized algorithm that, on the same input, determines the
height of every element of $(P, \dom)$ with expected total complexity $O(n w \log n)$.

The algorithms are analogous to Quicksort, and are based on a {\it
ternary} search tree, an extension of the well-known binary search
tree for maintaining elements of a linear order. A ternary search
tree for $(P, \dom)$, consists of a root, a left subtree, a
middle subtree and a right subtree. The root contains an element
$x \in P$ and the left, middle, and right subtrees are ternary
search trees for the restrictions of $(P, \dom)$ to the sets
$\{y\,|\,x \dom y\}$, $\{y\,|\,x \incomp y\}$ and $\{y\,|\, y \dom x\}$,
respectively. The ternary search tree for the empty poset consists
of a single empty node. A randomized algorithm to construct a
ternary search tree for $(P, \dom)$ assigns a random element of
$P$ to the root, compares each of the $n-1$ other elements to the
element at the root to determine the sets associated with the
three children of the root, and then, recursively, constructs a
ternary search tree for each of these three sets.

Define the weight of an internal node $x$ of a ternary search tree
as the total number of internal nodes in its three subtrees, and
the weight of a ternary search tree as the sum of the weights of
all internal nodes. Then the number of queries required to
construct a ternary search tree is exactly the weight of the tree.

\begin{theorem}
The expected weight of a ternary search tree for any  poset of
size $n$ and width $w$ is $O(n(\log n + w))$.
\end{theorem}

\begin{proof}[Proof Sketch] Consider the path from the root to a given
element $x$.
The number of edges in this path from a parent to a middle subtree is at most $w$, and the expected number of edges from a parent to a left or right subtree is $O(\log n)$ since, at every step along the path, the probability is at least $1/2$ that the sizes of the left and right subtrees differ by at most a factor of $3$. It follows that the expected contribution of any element to the weight of the ternary search tree is $w + O(\log n)$. \end{proof}

Once a ternary search tree for a poset has been constructed, a
linear extension can be constructed by a single depth-first
traversal of the tree. If $x$ is the element at the root, then the
linear extension is the concatenation of the linear extensions of
the following four subsets, corresponding to the node and its
three subtrees:  $\{y\,|\,x \dom y\}$, $\{x\}$,
$\{y\,|\,x \incomp y\}$ and
$\{y\,|\,y \dom x\}$.  The corollary below follows.

\begin{corollary}
There is a randomized algorithm of expected total complexity
$O(n(\log n + w))$ for computing a linear extension of a poset.
\end{corollary}

Let $h(x) = h$ be the height of element $x$ in $(P, \dom)$. Given a linear extension
$x_n > \cdots > x_2 > x_1$, it is easy to compute $h(x)$ for
each element $x$ by binary search, using the following observation: Let
$S(i,h)=\{x_j\,|\,j \leq i, h(x_j) = h \}$ be the set of elements of index at most $i$
in the linear extension and of height $h$ in $(P, \dom)$. Then
$|S(i,h)| \leq w$ (as the elements of $S(i,h)$ are pairwise incomparable),
and $h(x_{i+1}) > h$ if and only if there exists $x \in
S(i,h)$ such that $x_{i+1} \dom x$. Thus, given the sets $S(i,h)$,
for all $h$, we can determine $h(x_{i+1})$ and the sets $S(i+1,h)$,
for all $h$, in time $O(w \log i)$ using binary search. This
yields:

\begin{corollary}
Given a linear extension, there is a deterministic algorithm with
total complexity $O(wn \log n )$ to compute  the heights of all
elements of a partial order of size $n$ and width $w$. Combining
this algorithm with the above algorithm for computing a linear
extension, there is a randomized algorithm to determine the
heights of all elements with expected total complexity $O(wn \log n)$.
\end{corollary}

\section{Variants of the poset model}
\label{sec:posetvariants}
In this section, we discuss sorting in two variants of the poset model that occur when different
restrictions are relaxed.  First, we consider posets for which a bound
on the width is not known in advance.  Second, we allow the
irreflexivity condition to be relaxed, which leads to transitive relations.
We show that with relatively little overhead in complexity, sorting in either case
reduces to the problem of sorting posets.

\subsection{Unknown width}
\label{sec:unknownwidth}
Recall from Section~\ref{sec:sorting} that $\nposets$ is the number of posets of width at most $\wid$ on $n$ elements.

\begin{claim}
Given a set $P$ of $n$ elements and access to an oracle for poset $\po = (P,\dom )$
of unknown width $\wid$, there is an algorithm that sorts $P$ using
at most $2\log \wid\left(\log \nposetsAlt + 4\wid n\right) = \Theta(n\log \wid \left(\log n+\wid \right))$ queries, and there is an efficient algorithm that sorts $P$
using at most $8n\wid\log \wid \log (n/ (2\wid))$ queries with total complexity $O(n \wid^2 \log \wid \log (n/ \wid))$.
\end{claim}

\begin{proof}
We use an alternate version of {\sc EntropySort} that returns {\sc fail} if it cannot insert an element (while maintaining a decomposition of the given width) and an alternate version of {\sc Poset-Mergesort} that returns {\sc fail} if the {\sc Peeling} algorithm cannot reduce the size of the decomposition to the given width.
The first algorithm of the claim is, for $i=1,2,\ldots$, to run the alternate version of algorithm {\sc EntropySort} on input set $P$, the oracle, and width upper bound $2^i$, until the algorithm returns without failing.  The second algorithm is analogous but uses the alternate version of {\sc Poset-Mergesort}.  The claim follows from Theorems~\ref{th:entropysortanalysis} and~\ref{thm:mergesort}, and from the fact that we reach an upper bound of at most $2\wid$ on the width of $\po$ in $\log \wid$ rounds.
\end{proof}

\subsection{Transitive relations}
\label{sec:transrelns}

A partial order is a particular kind of transitive relation.  In fact,
our results generalize to the case of arbitrary transitive relations
(which are not necessarily irreflexive) and are therefore relevant to a broader set of applications.
Formally, a transitive relation
is a pair $(P, \domtr)$, where $P$ is a set of elements and $\domtr \subseteq P \times P$
is transitive.  The \emph{width} of a transitive relation is defined to be the maximum
size of a set of mutually incomparable elements.
We say that a poset $(P, \dom)$ is \emph{induced} by a transitive relation
$(P, \domtr)$ if $\dom \subseteq \domtr$.  A poset $(P, \dom)$ is \emph{minimally
induced} by $(P, \domtr)$ if for any relation $(x,y)\in \domtr \setminus \dom$,
the pair $(P, \dom \cup\, (x,y))$ is not a valid partial order, i.e. its corresponding graph contains a directed cycle.

We require the following lemma, bounding the width of a minimally induced poset.

\begin{lemma}
Let $(P, \dom)$ be a poset minimally induced by the transitive relation $(P, \domtr)$.
Then the width of $(P, \dom)$ is equal to the width of $(P, \domtr)$.
\end{lemma}

\begin{proof}
Suppose otherwise, that is, suppose that there is a pair of distinct elements $x, y\in P$ such that
$x \incomp y$ with respect to the partial order $(P, \dom)$, but $x$ and $y$ have some
relation in $(P, \domtr)$.  Without loss of generality, suppose that $x \domtr y$; it may
be simultaneously true that $y \domtr x$.
First, we note that $(P, \dom \cup\, (x,y))$ is a valid partial order; if it were not, i.e.
if the addition of $(x,y)$ introduced a cycle, then
it would be the case that $y \dom x$, which is a contradiction to their incomparability.
However, the poset $(P, \dom \cup\, (x,y))$ is also induced by $(P, \domtr)$,
which contradicts the assumption that $(P, \dom)$ is minimally induced.
\end{proof}

We denote by $\ora_\dom$ an oracle for a poset $(P,\dom)$ and by
$\ora_\domtr$ an oracle for a transitive relation $(P,\domtr)$.
In the following claim, we assume that the poset sorting algorithm outputs a chain decomposition (such
as a $\chainmerge$); if it does not, the total complexity of the algorithm for sorting a transitive relation increases a bit, but not its query complexity.

\begin{claim}
Suppose there is an algorithm $\A$ that, given a set $P$ of $n$ elements, access to an
oracle $\ora_\dom$ for a poset $\po=(P, \dom)$, and an upper bound of $\wid$ on the width of $\po$,
sorts $P$ using $f(n, \wid)$ queries and $g(n,\wid)$ total complexity.
Then there is an algorithm $\B$ that, given $P$, $\wid$, and access to an oracle $\ora_\domtr$ for a transitive
relation $(P, \domtr)$ of width at most $\wid$, sorts $P$ using $f(n,w) + 2n\wid$ queries and $g(n,w) + O(n\wid)$ total complexity.
\end{claim}

\begin{proof}
Given an oracle $\ora_{\domtr}$ for the transitive relation $(P, \domtr)$,
we define a special poset oracle $\ora$ that runs as follows:
Given a query $q(x,y)$, the oracle $\ora$ first checks if the relation between $x$ and $y$ can be inferred by transitivity and irreflexivity from previous responses.  If so, it outputs the appropriate inferred response; otherwise, it forwards the query to the oracle $\ora_{\domtr}$.  The oracle $\ora$ outputs the response of $\ora_{\domtr}$ except if both $x \domtr y$
and $y \domtr x$; in this case, $\ora$ outputs whichever relation is consistent with the partial order determined by previous responses (if both relations are consistent, then it arbitrarily outputs one of the two).  By definition, the responses of $\ora$ are consistent
with a partial order induced by $(P, \domtr)$.

The first step of algorithm $\B$ is to run algorithm $\A$ on input $P$ and $\wid$,
giving $\A$ access to the special oracle $\ora$, which $\B$ simulates using its access to $\ora_{\domtr}$.
Since $\A$ completely sorts its input, it reconstructs a poset induced
by $(P, \domtr)$ via $\ora$ that has a maximal set
of relations.  That is, there is a poset $\po=(P,\dom)$
minimally induced by $(P,\domtr)$ such that the
responses of $\ora$ to the sequence of queries made by $\A$ are
indistinguishable from the responses of $\ora_{\dom}$ to the same
sequence of queries.  Since $\po$ has the same width as $(P, \domtr)$,
it is valid to give $\A$ the upper bound of $\wid$.  Hence, $\A$ sorts $\po$
and outputs some chain decomposition $\C =\set{C_1, \ldots C_q}$ of $\po$ such that $q \leq \wid$.

The second step of algorithm $\B$ is to make a sequence of queries to the oracle $\ora_{\domtr}$
to recover the relations in $\domtr\setminus\dom$.  It
is similar to building a $\chainmerge$ data structure:
for all $i, j$, $1 \leq i, j \leq q$, for every element $x \in C_i$, we store the index of $x$ in chain $C_i$
and the index of the largest element $y \in C_j$ such that $x \domtr y$.  An analysis similar to the one for $\chainmerge$ (see Section~\ref{sec:chainmerge}) shows that it takes at most $2nq$ queries to the oracle $\ora_{\domtr}$ and $O(nq)$ total complexity to find all the indices.  The relation in $(P, \domtr)$ between any pair of elements can then be looked up in constant time.
\end{proof}


\begin{thebibliography}{1}

\bibitem{Brightwell} G.~Brightwell. ``Balanced Pairs in Partial Orders,''
{\em Discrete Mathematics} 201(1--3): 25--52, 1999.

\bibitem{Brightwell&Goodall:SortingLowerBound} G.~Brightwell and S.~Goodall. ``The Number
of Partial Orders of Fixed Width,'' {\em Order} 20(4): 333--345,
2003.

\bibitem{Brightwell&Winkler} G.~Brightwell and P.~Winkler. ``Counting Linear Extensions is
\#P-Complete,'' {\em STOC} 1991.

\bibitem{coverthomas} T.~M.~Cover and J.~A.~Thomas. {\em Elements of information
theory.} New York: John Wiley \& Sons Inc., 1991.

\bibitem{Faigle&Turan} U.~Faigle and Gy.~Tur\'{a}n. ``Sorting and Recognition Problems for Ordered Sets,''
{\em SIAM J. Comput.} 17(1): 100--113, 1988.

\bibitem{FordFulkerson} L.~R.,~Jr.,~Ford and D.~R.~Fulkerson. {\em Flows in Networks.} Princeton University Press, 1962.

\bibitem{Fredman} M.~Fredman. ``How good is the information theory bound in sorting?''
{\em Theor. Comput. Sci.} 1(4): 355-–361, 1976.

\bibitem{Fuss&Ga} F.~Fussenegger and H.~N.~Gabow. ``A Counting Approach to Lower
Bounds for Selection Problems,'' {\em Journal of the ACM} 26(2): 227--238, 1979.

\bibitem{Kahn&Kim} J.~Kahn and J.~H.~Kim. ``Entropy and Sorting,''
{\em STOC}, 178 -- 187, 1992.

\bibitem{Kahn&Saks} J.~Kahn and M.~Saks. ``Balancing poset extensions,''{\em Order} 1(2): 113--126,1984.

\bibitem{kislitsyn} S.~S.~Kislitsyn, ``A finite partially ordered set and its corresponding set of permutations,'' {\em Matematicheskie Zametki} 4(5): 511--518, 1968.

\bibitem{Knuth} D.~Knuth. {\em The Art of Computer Programming: Sorting and Searching}, Massachusetts: Addison-Wesley, 1998.

\bibitem{Linial} N.~Linial. ``The Information theoretic bound is good for merging,''
{\em SIAM J. Comput. SICOMP} 13(4): 795-–801, 1984.

\bibitem{Newman03} M. E. J. Newman, {\em SIAM Review} 45: 167--256, 2003.

\bibitem{OnakParys} K.~Onak and P.~Parys. Generalization of Binary Search:
Searching in Trees and Forest-Like Partial Orders. {\em FOCS} 2006.

\bibitem{trotter&felsner} W.~Trotter and S.~Felsner, ``Balancing pairs in partially ordered sets,'' {\em Combinatorics, Paul Erdos is Eighty} I: 145--157, 1993.

\end{thebibliography}
\end{document}